\pdfoutput=1
\documentclass[12pt,a4paper]{amsart}
\usepackage[a4paper,inner=2.5cm,outer=2.5cm,top=2.5cm,bottom=2.5cm]{geometry}
\usepackage{amsmath,amssymb,amsthm,enumerate,mathtools,stmaryrd
}
\usepackage{hyperref}
\hypersetup{colorlinks=true,linkcolor=blue,citecolor=teal,filecolor=magenta,urlcolor=cyan}

\usepackage[utf8]{inputenc}

\usepackage{makecell}
\usepackage{graphicx}

\usepackage{float}

\usepackage{extarrows}

\usepackage{xcolor}

\newcommand{\C}{\mathbb{C}}

\newcommand{\dtau}{\frac{\partial}{\partial \tau}}

\usepackage{url}

\usepackage[backend=bibtex,style=alphabetic,sorting=nyt,isbn=false,url=false,doi=true,maxalphanames=10,minalphanames=4,mincitenames=4,maxcitenames=10,minnames=4,maxnames=10,giveninits=true,maxbibnames=99]{biblatex}
\theoremstyle{plain}
\newtheorem{theorem}{Theorem}[section]

\newtheorem{corollary}[theorem]{Corollary}

\theoremstyle{definition}
\newtheorem{definition}[theorem]{Definition}

\makeatletter
\@addtoreset{proofpart}{theorem}
\makeatother
\theoremstyle{remark}
\newtheorem{remark}[theorem]{Remark}

\newcommand{\Z}{\mathbb{Z}}

\newcommand{\cS}{\mathcal{S}}

\newcommand{\cK}{\mathcal{K}}

\newcommand{\cW}{\mathcal{W}}
\newcommand{\cT}{\mathcal{T}}

\newcommand{\VEVc}[1]{{\big\langle \! 0 \big| {#1} \big| 0 \!\big\rangle^\circ}}
\newcommand{\vac}{{\big| 0\! \big\rangle}}
\newcommand{\covac}{{\big\langle \! 0 \big| }}

\newcommand{\res}{\mathop{\rm res}}

\newcommand{\restr}[2]{\mathop{\big\lfloor_{{#1}\to {#2}}}}
\newcommand{\set}[1]{\llbracket {#1} \rrbracket}

\newcommand{\OmClass}{{\mathit{\Omega}\!\!\!\!\mathit{\Omega}}}

\DeclarePairedDelimiter\floor{\lfloor}{\rfloor}

\newcommand{\lcoord}{\xi}

\title{Any topological recursion on a rational spectral curve is KP integrable}

\author[A.~Alexandrov]{A.~Alexandrov}
\address{A.~A.: Center for Geometry and Physics, Institute for Basic Science (IBS), Pohang 37673, Korea
}
\email{alex@ibs.re.kr}

\author[B.~Bychkov]{B.~Bychkov}
\address{B.~B.: Department of Mathematics, University of Haifa, Mount Carmel, 3498838, Haifa, Israel}
\email{bbychkov@hse.ru}

\author[P.~Dunin-Barkowski]{P.~Dunin-Barkowski}
\address{P.~D.-B.: Faculty of Mathematics, HSE University, Usacheva 6, 119048 Moscow, Russia; HSE--Skoltech International Laboratory of Representation Theory and Mathematical Physics, Skoltech, Bolshoy Boulevard 30 bld. 1, 121205 Moscow, Russia; and NRC “Kurchatov Institute” -- ITEP, 117218 Moscow, Russia}
\email{ptdunin@hse.ru}

\author[M.~Kazarian]{M.~Kazarian}
\address{M.~K.: Faculty of Mathematics, HSE University, Usacheva 6, 119048 Moscow, Russia; and Igor Krichever Center for Advanced Studies, Skoltech, Bolshoy Boulevard 30 bld. 1, 121205 Moscow, Russia}
\email{kazarian@mccme.ru}

\author[S.~Shadrin]{S.~Shadrin}
\address{S.~S.: Korteweg-de Vries Institute for Mathematics, University of Amsterdam, Postbus 94248, 1090GE Amsterdam, The Netherlands}
\email{S.Shadrin@uva.nl}	

\begin{document}
	
\begin{abstract}
We prove that for any initial data on a genus zero spectral curve the corresponding correlation differentials of topological recursion are KP integrable. As an application we prove KP integrability of partition functions associated via ELSV-type formulas to the $r$-th roots of the twisted powers of the log canonical bundles.
\end{abstract}
	
\maketitle
	
\setcounter{tocdepth}{3}
\tableofcontents

\section{Introduction}

Topological recursion (TR)~\cite{EO-1st} is an explicit recursive procedure that associates to a small input data, the so-called spectral curve data that consists of a Riemann surface, two functions on it, and a bi-differential, a system of the so-called correlation differentials $\{\omega^{(g)}_n\}_{g\geq 0, n\geq 1}$, which are symmetric $n$-differentials on the $n$-th Cartesian power of the underlying Riemann surface. We omit the explicit definition of TR along with its possible variations and ramifications, as it will not be directly used in the present paper and refer the reader to \emph{op.~cit.} as well as further surveys available in~\cite{EO,Eynard-Lecture}.

Topological recursion has proved to be a ubiquitous tool that interconnects combinatorics, enumerative geometry, and integrability, with further applications to knot invariants and free probability, just to name a few. For many results of the principal interest at the interface of these areas, the only available proofs go through topological recursion, the notable examples include e.g.,~\cite{BDKS2,DKPS-qr-ELSV,AS,CLPS}.

A lot of effort is put into establishing the integrability property of the correlation differentials, see~\cite{BorotEynard,ABDKS3,eynard2024hirota}. In particular, in~\cite{ABDKS3} we gave a definition of KP integrability for a family of differentials $\{\omega^{(g)}_n\}$ (not necessarily coming from topological recursion), proved that it requires the curve on which these differentials are defined to be of genus $0$, and used its compatibility with the so-called $x-y$ duality in topological recursion in order to establish KP integrability for a large class of examples, providing, in particular, the explicit way to derive the associated tau functions.

The main goal of this paper is to complete the study of KP integrability for the TR differentials as a property of a family of differentials. We prove that for any TR input data with the spectral curve of genus $0$ the associated differentials are KP integrable.

Note that, while this result is the most general that one might hope to obtain in the framework of our understanding of KP integrability as a global property of a system of differentials, it does not exhibit all interrelations between KP and TR: there are conjectures on integrability of the system of correlation differentials associated to higher genera spectral curves that involve corrections with $\Theta$-functions and is understood as the property of particular expansions of the differentials, see a recent survey in~\cite{eynard2024hirota}.

The rest of the introduction is organized as follows: in Section~\ref{sec:IntroMainTheorem} we proceed to formulate the main result of this paper, and in Section~\ref{sec:IntroBackground} we provide details on the results in~\cite{ABDKS3} briefly mentioned above; they form the necessary background for our main theorem.

\subsection{Main theorem} \label{sec:IntroMainTheorem}

Consider any spectral curve data $(\Sigma,x,y,B)$ such that
\begin{itemize}
\item $\Sigma$ is of genus zero, that is, $\Sigma=\C P^1$;
\item $B$ is the canonical Bergman kernel, that is, $B(z_1,z_2)=\frac{dz_1dz_2}{(z_1-z_2)^2}$, where $z$ is a global affine coordinate on $\Sigma=\C P^1$;
\item $dx$ is a meromorphic differential such that all its zeros are simple; we denote these zeros of $dx$ by $q_1,\dots,q_N$;
\item $y$ is given as a collection of formal power expansions at the points $q_i$ such that $dy(q_i)\ne 0$ for all $i=1,\dots,N$.
\end{itemize}
Note that the notation we use is the same as in the papers we referred to above for the definitions of TR, but we relax the conditions on $x$ and $y$ (which is very natural within the approach of Eynard in~\cite{Eynard-Intersection}). The topological recursion procedure associates with this initial data a system of symmetric $n$-differentials $\omega^{(g)}_n$ on $\Sigma^n$ for all $g\ge0$, $n\ge1$ such that
\begin{itemize}
\item$\omega^{(0)}_1=y\,dx$, $\omega^{(0)}_2=B$. Note that $\omega^{(0)}_1$ is only defined as a formal expansion at the points $\{q_1,\dots, q_N\}$.
\item For $2g-2+n>0$, the forms $\omega^{(g)}_n$ are global meromorphic, i.e.~rational; the only possible poles that they can possess (with respect to each argument) are at the zeros of~$dx$ only. Moreover, the order of poles is at most $2(3g-3+n)+2$ and the principal parts of the poles are odd with respect to the deck transformation of the function~$x$ considered as a local degree two ramified covering.
\end{itemize}
All these properties follow directly from~\cite{Eynard-Intersection}.

The main result of this paper is the following theorem that establishes integrable properties of all systems of correlation differentials obtained by topological recursion on the genus zero spectral curve.

\begin{theorem}\label{th:KP}
The system of TR differentials for the input data as above possesses KP integrability property.
\end{theorem}

The KP integrability as a property of a system of differentials is discussed below, see Definition~\ref{def:KPdiff}.
This theorem admits several generalizations. One of them is already included in the description of the input data above, as we assumed $y$ to be given as a collection of formal power expansions at the points $q_i$ rather than a globally defined function which is the standard setup for topological recursion. But in fact we prove that one can further relax this setup and assume that $dx$ is also defined just in the vicinity of the points $q_i$. The latter level of generality matches exactly the one assumed in~\cite{Eynard-Intersection}.

\begin{remark}
Note that we do not provide a detailed description of tau functions generated through topological recursion, but instead focus on proving the integrability. The general family of tau functions that originates from topological recursion is much more extensive than the specific family with known explicit description, say, in terms of asymptotic description of integrals, constructed in \cite{ABDKS3}.
\end{remark}

\subsection{Background} \label{sec:IntroBackground}
Here we mostly follow \cite{ABDKS3} where the inverse of Theorem \ref{th:KP} is proven, see Theorem \ref{th:KPinv} below.

Let $\Sigma$ be a smooth complex curve that we refer below as the \emph{spectral curve}. Consider a system of $n$-differentials $\omega^{(g)}_n$ on $\Sigma^n$ defined for all $g\ge0$, $n\ge1$. We do not assume that these differentials are related by any recursion. We assume that all $\omega^{(g)}_n$'s are symmetric and meromorphic with no poles on the diagonals for $(g,n)\ne(0,2)$, and $\omega^{(0)}_2$ is also symmetric and meromorphic but it has a second order pole on the diagonal with bi-residue~$1$.
To pass from a system of differentials to KP tau functions, we need the following definition:

\begin{definition}
	A point $o\in\Sigma$ is called \emph{regular} for the system of differentials $\{\omega^{(g)}_n\}$ if $\omega^{(g)}_n-\delta_{(g,n),(0,2)}\frac{d\lcoord_1d\lcoord_2}{(\lcoord_1-\lcoord_2)^2}$ is regular at $(o,\dots,o)\in\Sigma^n$ for all $(g,n)\in\Z_{\geq 0}\times\Z_{>0}$, $(g,n)\not=(0,1)$, where~$\lcoord$ is a local coordinate on $\Sigma$ near~$o$.
\end{definition}

\begin{remark}
	Note that the regularity condition for $\{\omega^{(g)}_n\}$ is independent of a choice of local coordinate. 	Note also that we have no condition for $\omega^{(0)}_1$.
\end{remark}

For a regular point~$o$, and an arbitrary local coordinate~$\lcoord$, the forms $\omega^{(g)}_n$ can be expanded into a series
\begin{equation}\label{eq:omega-expansion}
	\omega_n=\sum_{\substack{g\geq 0\\ 2g-2+n\geq 0}}\hbar^{2g-2+n}\omega^{(g)}_n=\delta_{n,2}\tfrac{d\lcoord_1d\lcoord_2}{(\lcoord_1-\lcoord_2)^2}+
	\sum_{k_1,\dots,k_n=1}^{\infty}f_{k_1,\dots,k_n}\prod_{i=1}^n k_i \lcoord_i^{k_i-1} d\lcoord_i,
\end{equation}
where the coefficients $f_{k_1,\dots,k_n}$ expand as $\sum_{g=0}^\infty \hbar^{2g-2+n} f^{(g)}_{k_1,\dots,k_n}$.
Introduce the associated potential $F=F_{o,\lcoord}$ as
\begin{equation}\label{eq:F-expansion}
	F(p_1,p_2,\dots)\coloneqq\sum_{\substack{g\geq 0, n\geq 1\\ 2g-2+n\geq 0}}\hbar^{2g-2+n}F^{(g)}_{n}=\sum_{n=1}^\infty\frac{1}{n!}
	\sum_{k_1,\dots,k_n=1}^{\infty}f_{k_1,\dots,k_n}{p}_{k_1}\dots {p}_{k_n},
\end{equation}
where $F^{(g)}_n$ are homogeneous formal series of degree~$n$ and ${p}_i$, $i=1,2,\dots$, are formal variables.

This way we associate to a given system of differentials $\omega^{(g)}_n$ a large family of potentials $F_{o,\lcoord}$: the freedom in its definition consists of a choice of a regular point $o$ and the choice of a local coordinate $\lcoord$ at $o$.
We would like to study conditions assuring that
\begin{equation}\label{eq:ZFrel}
	Z_{o,\lcoord}\coloneqq\exp(F_{o,\lcoord})
\end{equation} is a tau function of KP hierarchy.

\begin{theorem}\cite[Theorem 2.3]{ABDKS3}\label{thm:KPglobal} 
If $Z_{o,\lcoord}$ is a tau function of KP hierarchy for some choice of a regular point $o$ and a coordinate $\lcoord$ at this point, then it is a tau function of KP hierarchy for any other choice of a regular point and a local coordinate at that point.
\end{theorem}

This allows us to give the following definition, which introduces the KP integrability as an internal property of the system of differentials:

\begin{definition} \label{def:KPdiff} A system of differentials $\{\omega^{(g)}_n\}$ satisfying the assumptions as above is called KP-integrable, if the associated partition function $Z_{o,\lcoord}$ is a KP tau function for at least some choice of a regular point $o$ and a local coordinate $\lcoord$.
\end{definition}

\begin{remark} \label{rem:01irrelevant}
	The choice of $\omega^{(0)}_1$ does not affect the KP integrability property, it can be chosen in an arbitrary way. For definiteness, in all expressions which contain a summation over genus we do not include the terms with $(g,n)=(0,1)$.
\end{remark}

\begin{remark}
	A change of the local spectral parameter is known to provide a symmetry of the KP hierarchy \cite{Shiota}, given by an element of the Virasoro subgroup of $\widehat{GL}_\infty$.
\end{remark}


\begin{theorem}\cite[Theorem 3.8]{ABDKS3}\label{th:KPinv}
If the differentials $\{\omega^{(g)}_n\}$ produced by TR are KP integrable, then the spectral curve $\Sigma$ is rational.
\end{theorem}

{We refer the reader to \cite[Section 5]{ABDKS3} for the proof. In a nutshell, the analysis of the KP integrability in the semi-classical limit performed by Takasaki--Takebe in~\cite{TakTak} implies restrictions on the Bergman kernel. Namely, assume that we have a possibly non-compact Riemann surface $C$ and $\omega^{(0)}_2=B$ is a symmetric bi-differential with the only pole being the double pole along the diagonal. Then it is required that there exists a meromorphic function $z\colon C\to \mathbb{C}P^1$ that establishes the embedding of $C$ into $\mathbb{C}P^1$ such that $B$ is the pull-back of the standard Bergman kernel on $\mathbb{C}P^1$. Thus $C$ is either isomorphic to an open domain in $\mathbb{C}P^1$ or just to the whole $\mathbb{C}P^1$, and $B$ must be the restriction of the standard Bergman kernel defined on  $\mathbb{C}P^1$. Note that this argument does not involve any properties of $dx$ or $y$ or other differentials, and in particular, we do not really need an assumption that the differentials $\{\omega^{(g)}_n\}$ are produced by TR.} 

The upshot of Theorems~\ref{th:KP} and~\ref{th:KPinv} can be then formulated as follows:
\begin{corollary} A system of differentials $\{\omega^{(g)}_n\}$ produced by topological recursion  is KP integrable if and only if the spectral curve is rational.
\end{corollary}

\subsection{Applications} Theorem~\ref{th:KP} immediately has a variety of applications. Namely, for any input data of topological recursion with a rational spectral curve we can immediately state the KP integrability of the resulting system of differentials and the partition functions associated to their expansions. In order to illustrate this principle, we analyze the spectral curve $x=\log z-z^r$, $y=z^s$, $z\in \C P^1$, $r,s\in \Z_{>0}$, associated in~\cite{LPSZ} to the Chern classes of the twisted $r$-th roots of the tensor powers of the log canonical bundle, the so-called $\OmClass$-classes, or, in other terminology, the Chiodo classes. 

In some cases the KP integrability of the associated partition functions was known --- for some special choices of $r$, $s$, and local expansions one obtains instances of hypergeometric KP tau functions~\cite{DKPS-qr-ELSV,BDKS2} associated to the so-called orbifold $r/s$-spin Hurwitz numbers (the ratio $r/s$ is required to be integer). In this paper we generalize these results to arbitrary positive integers $r$ and $s$, see Corollary~\ref{cor:OmegaKP}.

\subsection{Organization of the paper} In Section~\ref{sec:GlobalKPSymmetries} we develop the theory of KP symmetries for the notion of KP integrability considered as a global property of a system of differentials.
In Section~\ref{sec:ProofOfTheorem} we use the theory of KP symmetries in combination to the deformation formulas for the differentials of topological recursion in order to reduce the KP integrability of a given system of TR differentials to the know instances of KP integrability proved in~\cite{ABDKS3}, which establishes our main theorem, Theorem~\ref{th:KP}. Section~\ref{sec:Generalized} is devoted to a generalized setup where we no longer assume that function $x$ is globally defined. Finally, in Section~\ref{sec:ApplicationRRoots} we provide an application to the KP integrability of the $\OmClass$-classes. 

In Appendix~\ref{sec:varTR} we provide some background on the deformation formulas for the differentials obtained by topological recursion.

\subsection{Notation} Throughout the text we use the following notation:
\begin{itemize}
	\item $\set{n}$ denotes $\{1,\dots,n\}$.
	\item $z_I$ denotes $\{z_i\}_{i\in I}$ for $I\subseteq \set{n}$.
	\item $[u^d]$ denotes the operator that extracts the corresponding coefficient from the whole expression to the right of it, that is, $[u^d]\sum_{i=-\infty}^\infty a_iu^i \coloneqq a_d$.
	\item
	$\restr{u}{v}$ denotes the operator of substitution (or restriction), that is, $\restr{u}{v} f(u) \coloneqq f(v)$.
\end{itemize}

\subsection{Acknowledgments} A. A. was supported by the Institute for Basic Science (IBS-R003-D1). B.B. was supported by the ISF Grant 876/20. B.B., P.D.-B., and M.K. were supported by the Russian Science Foundation (grant No. 24-11-00366). S.~S. was supported by the Netherlands Organization for Scientific Research.

{We thank the referees for useful remarks.}
 
\section{Global KP symmetries} \label{sec:GlobalKPSymmetries}

For a given system of meromorphic symmetric differentials $\omega^{(g)}_n$ and a function~$x$ such that $dx$ is meromorphic and $\omega^{(g)}_n-\delta_{(g,n),(0,2)}\frac{dx_1dx_2}{(x_1-x_2)^2}$ is regular on the diagonals we set
\begin{equation} \label{eq:omega-n}
\omega_n=\sum_{\substack{g\geq 0\\ 2g-2+n\geq 0}}\hbar^{2g-2+n}\omega^{(g)}_n
\end{equation}
and define
\begin{align}
\omega^{(0),{\rm sing}}_2(z_1,z_2)&=\frac{dx(z_1)\,dx(z_2)}{(x(z_1)-x(z_2))^2},
\\\cT_n(z^+,z^-;z_{\set n})&=\sum_{k=1}^\infty\frac1{k!}
\int\limits_{z^-}^{z^+}\!\!\dots\!\!\int\limits_{z^-}^{z^+}
\bigl(\omega_{n+k}(\underbrace{\cdot,\dots,\cdot}_k,z_{\set n})-\delta_{(n,k),(0,2)}\omega^{(0),{\rm sing}}_2\bigr),
\\
\label{eq:Omega}
\Omega_n(z^+,z^-;z_{\set n})&=\sqrt{\omega^{(0),{\rm sing}}_2(z^+,z^-)}\;
e^{\cT_0(z^+,z^-)}\sum_{\set n=\bigsqcup\limits_\alpha J_\alpha,~J_\alpha\ne\emptyset}
\prod_\alpha \cT_{|J_\alpha|}(z^+,z^-;z_{J_\alpha}),
\\\label{eq:cW}
\cW_n(x,u;z_{\set n})&=\restr{z^{\pm}}{e^{\pm \frac{u\hbar}{2}\partial_{x}}z}\Omega_n(z^+,z^-;z_{\set n}).
\end{align}

\begin{remark}\label{rem:bW-polynomiality}
The arguments $(z^+,z^-)$ of $\Omega_{n}$ are treated as a point of the Cartesian square $\Sigma^2$ and are considered in a certain vicinity of the diagonal $z^+=z^-$. The arguments $(z,u)$ of $\cW_n$ provide a particular parametrization of this vicinity through the substitutions
\begin{equation}
z^{+}=e^{\frac{u\hbar}{2}\partial_{x}}z,\qquad
z^{-}=e^{-\frac{u\hbar}{2}\partial_{x}}z.
\end{equation}
\end{remark}

\begin{remark}
The half-differentials $\sqrt{dx(z^\pm)}$ entering $\Omega_n$ are understood formally. They attain a more definite meaning only after the above substitution. Namely, we have
\begin{equation}
\restr{z^{\pm}}{e^{\pm \frac{u\hbar}{2}\partial_{x}}z} x^\pm=x\pm\frac{u\hbar}{2},
\qquad\restr{z^{\pm}}{e^{\pm \frac{u\hbar}{2}\partial_{x}}z}
\sqrt{\omega^{(0),{\rm sing}}_2(z^+,z^-)}=\frac{dx}{u\hbar}.
\end{equation}
A reason to introduce $\Omega_n$ in this form is that they are actually independent of $x$ and are uniquely determined by the differentials $\omega^{(g)}_n$. Indeed, a different choice of the $x$-function changes both the factor $\sqrt{\omega^{(0),{\rm sing}}_2(z^+,z^-)}=\frac{\sqrt{dx^+dx^-}}{x^+-x^-}$ and the regularization term $\omega^{(0),{\rm sing}}_2$ of the singular form $\omega^{(0)}_2$ under the integral in the definition of $\cT_0$. An easy computation shows that these two changes compensate one another, see~\cite[Equation (3.10)]{ABDKS3}.
\end{remark}

\begin{remark}\label{rmk3.5} We regard $\cW_n$ as an $(n+1)$-differential on $\Sigma^{n+1}$ depending on an additional parameter $u$.
One can see that in the definition of  $\cW_n$ the integration can be interpreted as an (infinite)  combination of the differentials $\omega^{(\tilde g)}_{\tilde n}$ and their derivatives. This fact follows from the equality
\begin{equation}
\restr{\tilde z^{\pm}}{e^{\pm \frac{u\hbar}{2}\partial_{x}}z}
\;\int\limits_{z^-}^{z^+}\!\!\dots\!\!\int\limits_{z^-}^{z^+}
\omega_{n+k}(\underbrace{\cdot,\dots,\cdot}_k,z_{\set n})
=\prod_{i=1}^{k}\bigl(\restr{\tilde z_i}{z} u\hbar\,\cS(u\hbar\,\partial_{\tilde x_i})\bigr)\;
\frac{\omega_{n+k}(\tilde z_{\set{k}},z_{\set n})}
{\prod_{i=1}^{k} d\tilde x_i},
\end{equation}
where $\cS(u)=u^{-1}(e^{u/2} - e^{-u/2})$.
As a corollary, the coefficient of each monomial in~$\hbar$ and~$u$ in $\cW_n$ is a differential polynomial in $\omega$'s. In particular, let us set $\cW^{(g),r}_{n}\coloneqq[u^{r}\hbar^{2g-1+n}]\cW_n$. Then, we have
\begin{align}\label{eq:W0}
\cW^{(g),0}_n(z;z_{\set n})&=\omega^{(g)}_{n+1}(z,z_{\set n}),
\\\label{eq:W1}
\cW^{(g),1}_n(z;z_{\set n})&=\frac{1}{2\,dx(z)}
\Bigl(\omega^{(g-1)}_{n+2}(z,z,z_{\set n})+
 \!\!\!\sum_{\substack{g_1+g_2=g\\I_1\sqcup I_2=\set{n}\\ 2g_i-1+|I_i| \geq 0 }}\!\!\!\omega^{(g_1)}_{|I_1|+1}(z,z_{I_1})\,\omega^{(g_2)}_{|I_2|+1}(z,z_{I_2})
 \Bigr).
\end{align}
\end{remark}

\begin{remark} Note that in~\eqref{eq:omega-n} we used our convention that the sum over genus excludes the $(g,n)=(0,1)$ term, which affects~\eqref{eq:Omega} and~\eqref{eq:cW}. In some other situations, for instance, in the study of the loop equations and/or deformations of the differentials, these terms should not be excluded. In particular, below we will also need a version of $\cW^{(g),1}_n$ with the $(0,1)$ terms included, which denote by
	\begin{align}\label{eq:tildeW1}
		\widetilde{\cW}^{(g),1}_n(z;z_{\set n})&=\frac{1}{2\,dx(z)}
		\Bigl(\omega^{(g-1)}_{n+2}(z,z,z_{\set n})+
		\!\!\!\sum_{\substack{g_1+g_2=g\\I_1\sqcup I_2=\set{n} }}\!\!\!\omega^{(g_1)}_{|I_1|+1}(z,z_{I_1})\,\omega^{(g_2)}_{|I_2|+1}(z,z_{I_2})
		\Bigr).
	\end{align}
\end{remark}

\begin{theorem}\label{thm:Delta-omega}
For $n\geq 1$ the transformation
\begin{equation}\label{eq:Delta-omega1}
\Delta\omega_n(z_{\set n})=\Omega_n(z^+,z^-;z_{\set n})
\end{equation}
is an infinitesimal KP symmetry for a system of differentials $\{\omega^{(g)}_n\}$ for any specialization of $z^+$ and $z^-$. In other words, the above formula provides a family of infinitesimal KP symmetries with $(z^+,z^-)\in\Sigma^2$ regarded as parameters of the family.
\end{theorem}

By an infinitesimal transformation here we mean $\omega_n(z_{\set n})\mapsto \omega_n(z_{\set n})+\epsilon\, \Delta\omega_n(z_{\set n})$ with an infinitesimal parameter $\epsilon$.
By a specialization $z^+$ and $z^-$ we mean a particular choice of these points in a vicinity of the diagonal in~$\Sigma^2$, or the coefficient of any monomial in $z^+,z^-$ in the power expansion in the local coordinate~$z$ at some point of~$\Sigma$, or the integral over some contour, or the residue at some point of the spectral curve.

More explicitly, the statement of Theorem \ref{thm:Delta-omega} means the following. Assume that the differentials $\omega^{(g)}_n$ depend on an additional parameter~$\tau$ and the dependence of the differentials on~$\tau$ is governed by the differential equation of the form
\begin{equation}\label{eq:Delta-omega}
\frac{d}{d\tau}\omega_n(z_{\set n})=\Omega_n(z^+,z^-;z_{\set n})
\end{equation}
for some specialization of $z^+,z^-$ also possibly depending on~$\tau$ and $\hbar$. In what follows we consider the deformations with $z^+,z^-$ having at most simple poles in $\hbar$, thus the right hand side of~\eqref{eq:Delta-omega1} might contain a pole in $\hbar$ as well; this pole should be deducted as it does not affect the KP integrability and is removed on the left hand side by our conventions, cf.~Remark~\ref{rem:01irrelevant}.  Assume also that the differentials possess KP integrability property for some value of~$\tau$. Then they possess this property for all values of~$\tau$.

\begin{corollary}\label{cor:Delta-omega}
For any $r$ and any specialization of $z\in\Sigma$ the transformation
\begin{equation}
\Delta\omega^{(g)}_n(z_{\set n})=\cW^{(g),r}_n(z;z_{\set n})
\end{equation}
is an infinitesimal KP symmetry for a system of differentials $\{\omega^{(g)}_n\}$.
\end{corollary}

\begin{proof}[Proof of Theorem~\ref{thm:Delta-omega}]
First, assume that $z^+$ and $z^-$ are chosen in a neighborhood of some point $o\in\Sigma$ regular for all differentials $\omega^{(g)}_n$. Let us choose an arbitrary local coordinate~$z$ at this point and identify a point in a neighborhood of~$o$ with the value of the coordinate~$z$ at this point. A choice of a point~$o$ and a coordinate~$z$ determines the corresponding potential~$F_{o,z}=F_{o,z}(p_1,p_2,\dots)$ defined by \eqref{eq:omega-expansion} and \eqref{eq:F-expansion}.

Now we are going to apply the logic that has earlier appeared in~\cite{ABDKS1,ABDKS3}, namely, we perform local computation for the formal expansions using the Fock space semi-infinite wedge formalism (we refer the reader to the \emph{op. cit.} for the notation that we use here, see, e.g.,~\cite[Section 2.2]{ABDKS3}), but the formulas that we obtain have the global meaning on the spectral curve. {Qua notation, we identify the Fock space with $\C[[p_1,p_2,\dots]]$, the operators $J_k$ are the multiplication by $p_{-k}$ for $k<0$ and derivatives $k\partial_{p_{k}}$ for $k>0$, and $\vac = 1$ and $\covac$ extracts the constant term.} 

The potential $F_{o,z}$ defines, in turn, the corresponding vector $\hat Z_{o,z}\vac$ in the Fock space where $\hat Z_{o,z}=e^{F_{o,z}(J_{-1},J_{-2},\dots)}$. With this notation, the correspondence between the potential and $n$-differentials takes the form of connected vacuum expectation values
\begin{equation}
\omega_{n}(z_{\set n})=\VEVc{J(z_1)\dots J(z_n)\hat Z_{o,z}}
\prod_{i=1}^n\frac{dz_i}{z_i},
\qquad J(z)=\sum_{m=-\infty}^\infty
z^mJ_{m}.
\end{equation}

The definition of~$\Omega_n$ specialized to the setting of power expansions in~$z$ variables can be represented, in turn, in the form
\begin{equation}
\Omega_{n}(z^+,z^-;z_{\set n})-\delta_{n,0}\tfrac{\sqrt{dz^+\,dz^-}}{z^+-z^-}=\VEVc{J(z_1)\dots J(z_n)\cK(z^+,z^-)\hat Z_{o,z}}\sqrt{dz^+dz^-}
\prod_{i=1}^n\frac{dz_i}{z_i},
\end{equation}
where $\cK$ is the operator acting in the Fock space and defined by
\begin{equation}
\cK(z_1,z_2)=\frac{e^{\sum\limits_{k<0}\frac{z_1^k-z_2^k}{k}J_k}e^{\sum\limits_{k>0}\frac{z_1^k-z_2^k}{k}J_k}-1}{z_1-z_2}.
\end{equation}
It follows that the infinitesimal transformation of the vector $\hat Z_{o,z}\vac$ in the Fock space corresponding to the transformation~\eqref{eq:Delta-omega} is given by
\begin{equation}\label{eq:Delta-Z}
\Delta \hat Z_{o,z}\vac=\sqrt{dz^+dz^-}\;\cK(z^+,z^-)\hat Z_{o,z}\vac.
\end{equation}
It is known~\cite{MJD,Kac,BDKS1} that the operator $\cK(z_1,z_2)$ belongs to the Lie algebra $\widehat{\mathfrak{gl}}(\infty)$ of infinitesimal KP symmetries. 

We conclude that the transformations~\eqref{eq:Delta-Z} and~\eqref{eq:Delta-omega} preserve infinitesimally KP integrability for the potential~$F_{o,z}$ associated with the point~$o$ and the local coordinate~$z$ at this point (the factor $\sqrt{dz^+dz^-}$ is a numerical constant that can be set arbitrarily and does not affect KP integrability).
On the other hand, by Theorem~\ref{thm:KPglobal}, the KP integrability does not depend on a choice of the expansion point and the local coordinate. Hence, the transformation~\eqref{eq:Delta-omega} provides an infinitesimal KP symmetry for the whole globally defined system of differentials~$\{\omega^{(g)}_n\}$. Since these arguments are applied to arbitrary choice of~$o$, we obtain that the statement of Theorem~\ref{thm:Delta-omega} is valid in its full generality.
\end{proof}

\section{Proof of Theorem~\ref{th:KP}} \label{sec:ProofOfTheorem}

Let us study the dependence of the TR differentials on a choice of the function~$y$. Thus we consider a family of spectral curve data such that the spectral curve~$\Sigma$, the Bergman kernel~$B$, and the function $x$ are fixed (and hence, the positions of the points $q_1,\dots,q_N$ are also fixed) while the function $y$ (that is, its expansion at the points $q_i$) varies in a family smoothly depending on a parameter~$\tau$.
The corresponding dependence of the TR differentials on~$\tau$ in the stable range $2g-2+n>0$ is described by the equation~\cite{EO-1st}:
\begin{equation}\label{eq:var-of-y}
\frac{d}{d\tau}\omega^{(g)}_n(z_{\set n})=\sum_{i=1}^N\res\limits_{z=q_i}\left(
\omega^{(g)}_{n+1}(z_{\set n},z)\int_{q_i}^z\frac{d y}{d\tau}\,dx \right).
\end{equation}
This equation holds true for any variation of~$y$ such that $dy(q_i)\ne 0$; in particular, we allow a change of the constants $y'(q_i)$ as long as these constants are not vanishing in the family.

By Corollary~\ref{cor:Delta-omega} and~\eqref{eq:W0}, the right hand side of Eq.~\eqref{eq:var-of-y} is an infinitesimal KP symmetry. It follows that the KP integrability property is preserved in the family, and hence the system of differentials $\{\omega^{(g)}_n\}$ is KP integrable for some choice of~$y$ if and only if it is KP integrable for any other choice of~$y$, provided that $dy$ is non-vanishing at $q_1,\dots,q_N$.

Thus, in the setting of Theorem~\ref{th:KP}, it is sufficient to prove KP integrability in the case when $y=z$ is a global affine coordinate on $\Sigma=\C P^1$. The KP integrability for this choice of~$y$ {and $\Sigma$} is proved in~\cite{ABDKS3}. This completes the proof of Theorem~\ref{th:KP}.

\section{Generalized setting}

\label{sec:Generalized}

The statement of Theorem~\ref{th:KP} holds in a more general setting. It is sufficient to assume that $x$ is defined in a vicinity of the points $q_1,\dots,q_N$ such that $dx$ has simple zeros in these points and does not extend necessarily to the whole spectral curve, and  $dy$ is holomorphic and nonvanishing at these points. Even if $dx$ extends beyond this vicinity and has other zeros, they do not contribute to TR, by definition.

The definition of TR differentials can be extended to this relaxed setting as well. If the Bergman kernel $B$ entering the initial data of recursion is defined locally in a neighborhood of the points $q_j$, then the resulting differentials are also defined locally in a neighborhood of those points. If $\Sigma$ is compact and $B$ extends globally as a meromorphic bidifferential with no singularities beyond the diagonal, then the TR differentials $\omega^{(g)}_n$ are also defined globally for $2g-2+n>0$, and have no poles different from $q_1,\dots,q_N$ (with respect to each of the arguments), even though the functions $x$ and $y$ of the initial data are defined locally near the points~$q_j$.

\begin{theorem}\label{th:KP-generalized}
If the spectral curve is rational, $\Sigma=\C P^1$, and $B$ is the standard Bergman kernel, then the KP integrability for the TR differentials holds in the relaxed setting of this section as well.
\end{theorem}

\begin{proof}
Consider a variation of the initial data of recursion such that the spectral curve {$\Sigma=\C P^1$} and the {(standard)} Bergman kernel are fixed while the functions~$x$ and~$y$ vary in a one-parameter family depending on a parameter~$\tau$. We allow the positions of the points~$q_j$ to change in the family provided that they remain pairwise distinct and the assumptions on $dx$ and $dy$ are satisfied for all parameter values. Then, the dependence of the TR differentials on~$\tau$ in the stable range $2g-2+n>0$ is described by the following equation, see~\cite{K} and Appendix~\ref{sec:varTR}
\begin{equation}\label{eq:omega-B-variation}
\frac{d}{d \tau}\omega^{(g)}_{n}(z_{\set n})=\sum_{j=1}^N\res\limits_{z=q_j}\Bigg(
\cW^{(g),0}_n(z;z_{\set n})\int\limits^z\!\!\tfrac{d}{d \tau}\bigl(y\,dx\bigr)
-\widetilde{\cW}^{(g),1}_n(z;z_{\set n})\tfrac{d}{d \tau} x(z)\Bigg),
\end{equation}
where $\cW^{(g),0}_n$ and $\widetilde{\cW}^{(g),1}_n$ are defined by~\eqref{eq:W0} and~\eqref{eq:tildeW1}, respectively. By Theorem~\ref{thm:Delta-omega} and its Corollary~\ref{cor:Delta-omega}, the KP integrability for the TR differentials is preserved in the family (note that we can use $\widetilde{\cW}^{(g),1}_n$ instead of ${\cW}^{(g),1}_n$ once the singular term in $\hbar$ is removed). Any collection of initial data can be connected by an analytic family with those satisfying the assumptions of Theorem~\ref{th:KP}, that is, when $dx$ extends globally as a meromorphic differential with no zeros except $q_1,\dots,q_N$. Say, it can be connected with the case when $x(z)$ is a polynomial of degree $N+1$ with no multiple critical points. Therefore, the assertion of Theorem~\ref{th:KP-generalized} follows from that one of Theorem~\ref{th:KP}.

\end{proof}

\section{Application:  \texorpdfstring{$r$}{r}-th roots}

\label{sec:ApplicationRRoots}

Consider the moduli space of generalized $r$-spin structures $\overline{\mathcal{M}}_{g;\vec a}^{r,s}$ which parametrizes $r$-th roots of the line bundle
\begin{align}
	\omega_{\log}^{\otimes s}\biggl(-\sum_{i=1}^n a_i \mathsf{x}_i \biggr)
\end{align}
on genus $g$ curves with $n$ marked points $(C,\mathsf{x}_1,\dots,\mathsf{x}_n)$, where $\omega_{\log} = \omega(\sum_i \mathsf{x}_i)$, $r$ and $s$ are positive integers, and $a_1,\dots,a_n$ are integers satisfying
\begin{align}
	a_1 + a_2 + \cdots + a_n \equiv (2g-2+n)s \pmod{r}.
\end{align}
This condition guarantees the existence of a line bundle on the underlying curve whose $r$-th tensor power is isomorphic to $\omega_{\log}^{\otimes s}(-\sum_i a_i x_i)$. Let $\pi \colon {\mathcal{C}}_{g,\vec a}^{r,s} \to \overline{\mathcal{M}}_{g,\vec{a}}^{r,s}$ be the universal curve, and $\mathcal{L} \to \overline{\mathcal C}_{g,\vec a}^{r,s}$ the universal $r$-th root. The Chiodo formula expresses the Chern characters of the derived pushforward of the universal $r$-th root $\mathrm{ch}_m(R^{\bullet} \pi_{\ast}{\mathcal L})$  in terms of the tautological classes, see~\cite{Chi08}.

We can then consider the family of Chern classes pushed forward to the moduli spaces of stable curves under the natural projection map $\epsilon\colon  \overline{\mathcal{M}}_{g,\vec{a}}^{r,s} \to \overline{\mathcal{M}}_{g,n}$.
\begin{equation}
	\OmClass_{g,n}(r,s;a_1,\dots,a_n)
	\coloneqq
	\epsilon_{\ast}  \exp{\Biggl(
		\sum_{m=1}^\infty (-1)^m  (m-1)! \, \mathrm{ch}_m(R^{\bullet} \pi_{\ast}{\mathcal L})
		\Biggr)}
	\in
	R^*(\overline{\mathcal{M}}_{g,n}).
\end{equation}
These classes are called the $\OmClass$-classes, or, alternatively, the Chiodo classes in the literature.

Consider the spectral curve data
\begin{align} \label{eq:spectral-curve-data}
\Sigma=\C P^1, \quad  x=\log z - z^r, \quad y=z^s, \quad B=\frac{dz_1dz_2}{(z_1-z_2)^2},
\end{align}
where, as usual, $z$ is the global affine coordinate on $\C P^1$. The following theorem is established in~\cite{LPSZ} (see also~\cite{SSZ} for the $s=1$ case):

\begin{theorem}[\cite{LPSZ}]\label{TH1}
	The expansions of the correlation differentials of the spectral curve~\eqref{eq:spectral-curve-data} at the point $z_1=\cdots=z_n=0$ in the local coordinates $X_i=X(z_i)=\exp(x(z_i))$ are equal to
	\begin{align}\label{eq:correlation-forms}
		\omega^{(g)}_{n}= & d_1\otimes \cdots \otimes d_n \sum_{\mu_1,\dots,\mu_n=1}^\infty \prod_{j=1}^n \Bigg(X_j^{\mu_j}
		\frac{\left(\frac{\mu_j}{r}\right)^{\floor*{\frac{\mu_j}r}}}{\floor*{\frac{\mu_j}r}!}\Bigg)
		\times
		\frac{r^{2g-2+n+\frac{(2g-2+n)s+\sum_{j=1}^n \mu_j}r}}{s^{2g-2+n}}
		\\ \notag &
		\times
		\int_{\overline{\mathcal{M}}_{g,n}} \frac{\OmClass_{g,n} \left (r,s; r-r \left < \frac{\mu_1}{r} \right >,
			\dots,
			r - r \left < \frac{\mu_n}{r} \right > \right )
		}{\prod_{j=1}^n (1-\frac{\mu_j}{r}\psi_j)},
	\end{align}
	where $\frac{\mu}r = \lfloor \frac{\mu}r \rfloor + \langle \frac{\mu}r \rangle$ is the decomposition into the integer and the fractional parts.
\end{theorem}

Note that the spectral curve data satisfies the conditions of Theorem~\ref{th:KP}. Hence, we obtain the following immediate corollary:

\begin{corollary} \label{cor:OmegaKP} The following partition function is a tau function of KP hierarchy:
	\begin{align}
		\exp \Bigg(&\sum_{\substack{g\geq 0, n\geq 1 \\ 2g-2+n>0}} \frac{\hbar^{2g-2+n}}{n!} \sum_{\mu_1,\dots,\mu_n=1}^\infty \prod_{j=1}^n \bigg({p}_{\mu_j} \frac{\left(\frac{\mu_j}{r}\right)^{\floor*{\frac{\mu_j}r}}}{\floor*{\frac{\mu_j}r}!}\bigg) \times 		\frac{r^{2g-2+n+\frac{(2g-2+n)s+\sum_{j=1}^n \mu_j}r}}{s^{2g-2+n}}
		\\ \notag & \qquad \qquad \qquad \qquad
		\times
			\int_{\overline{\mathcal{M}}_{g,n}} \frac{\OmClass_{g,n} \left (r,s; r-r \left < \frac{\mu_1}{r} \right >,
			\dots,
			r - r \left < \frac{\mu_n}{r} \right > \right )
		}{\prod_{j=1}^n (1-\frac{\mu_j}{r}\psi_j)}
		\\ \notag  & + \frac 12 \sum_{\mu_1,\mu_2=1}^\infty {p}_{\mu_1}{p}_{\mu_2} [X_1^{\mu_1}X_2^{\mu_2}] \log\bigg(\frac{z_1-z_2}{x_1-x_2} \bigg)  \Bigg).
	\end{align}
\end{corollary}

\begin{remark} In the case $r=st$ for some positive integer $t$, this function up to an extra factor of $\exp(\hbar^{-1}\sum_{\mu} {p}_\mu [X^\mu] \int_0^z z^{s-1} (1-rz^r) dz)$ is a hypergeometric KP tau function studied in~\cite{DKPS-qr-ELSV,BDKS2}. So, in these cases the result of Corollary~\ref{cor:OmegaKP} was known.
\end{remark}

\appendix

\section{Variation of TR differentials}

\label{sec:varTR}

In this Appendix, we review formulas of~\cite{EO-1st,DOSS,K} describing variation of topological recursion differentials under a deformation of initial data of recursion.

Assume that the spectral curve data $(\Sigma,x,y,B)$ of topological recursion (or, more generally, the local spectral curve data as in Section~\ref{sec:Generalized}) varies in a family depending analytically on one additional parameter~$\tau$. Since zeros $q_j$ of $dx$ are nondegenerate, this property is preserved under variation of the spectral curve data, and by implicit function theorem, the dependence of positions of the points $q_j$ on the parameter is analytic. We assume that no zeros of $dx$ other than $q_1,\dots,q_N$ appear for the deformed spectral curve, or that additional zeros are ignored and do not contribute to the recursion, by definition. Then one may expect that the variation of the differentials~$\omega^{(g)}_{n}$ is determined by the variation of the functions $x$, $y$, and the bidifferential $B$ in a neighborhood of the points $q_j$. In fact, this is the case and we review here explicit formulas for the variation of $\omega^{(g)}_{n}$.

In order to give a meaning to the derivative of $\omega^{(g)}_{n}$ over parameter~$\tau$ we need to identify spectral curves for different parameter values, at least locally near the points~$q_j$. Let us take an arbitrary meromorphic function $\zeta$ on the spectral curve and denote by $\tfrac{d}{d \tau}\big|_\zeta$ the derivative at fixed~$\zeta$, that is, the partial derivative over~$\tau$ computed in any local coordinate on the spectral curve such that the expression for~$\zeta$ in this local coordinate is independent of~$\tau$.

Denote
\begin{align}\label{eq:omegast}
\omega^{(g),{\rm st}}_{n,2}(z,\tilde z;z_{\set n})&=
\frac12\biggl(\omega^{(g-1)}_{n+2}(z,\tilde z,z_{\set{n}})
+\!\!\!\!\sum\limits_{\substack{g_1+g_2=g,~
I_1\sqcup J_2=\set{n}\\\text{no }(0,1)\text{ or }(0,2)\text{ terms}}}\!\!\!\!
\omega^{(g_1)}_{|I_1|+1}(z,z_{I_1})\,\omega^{(g_2)}_{|I_2|+1}(\tilde z,z_{I_1})\biggr),
\\
\widetilde{\cW}^{(g),1}_{n}(z;z_{\set n})&=
\frac1{2\,dx(z)}\biggl(\omega^{(g-1)}_{n+2}(z,z,z_{\set{n}})
+\!\!\!\!\sum\limits_{\substack{g_1+g_2=g\\ I_1\sqcup J_2=\set{n}}}\!\!\!\!
\omega^{(g_1)}_{|I_1|+1}(z,z_{I_1})\,\omega^{(g_2)}_{|I_2|+1}(z,z_{I_1})\biggr)
\\ \notag &=
\frac1{dx(z)}\biggl(
\omega^{(g),{\rm st}}_{n,2}(z,z;z_{\set n})
 +\omega^{(g)}_{n+1}(z,z_{\set n})\,y(z)\,dx(z) \\ \notag & \qquad \qquad \ +\sum_{i=1}^n\omega^{(g)}_{n}(z,z_{\set{n}\setminus\{i\}})\,B(z,z_i)
 \biggr).
\end{align}

\begin{theorem}[\cite{K}]\label{thm:kaz}
For $2g-2+n>0$ the dependence of TR differentials $\omega^{(g)}_{n}$ on the parameter~$\tau$ is described by the equation
\begin{multline}\label{eq:xydef}
\frac{d}{d \tau}\Big|_\zeta\omega^{(g)}_{n}(z_{\set n})=\sum_{j=1}^N\res\limits_{z=q_j}\Bigg(
\omega^{(g)}_{n+1}(z,z_{\set n})\int\limits^z\!\!\tfrac{d}{d \tau}\big|_\zeta\bigl(y\,dx\bigr)-
\widetilde{\cW}^{(g),1}_n(z;z_{\set n})\tfrac{d}{d \tau}\big|_\zeta x(z)\Bigg)
 \\+\sum_{i,j=1}^N\res_{z=q_i}\res_{\tilde z=q_j}
 \omega^{(g),{\rm st}}_{n,2}(z,\tilde z,z_{\set n})\int\limits^z\!\!\!\int\limits^{\tilde z}\!\!\tfrac{d}{d \tau}\big|_\zeta B
 +\sum_{i=1}^n\sum_{j=1}^N\res\limits_{z=q_j}\omega^{(g)}_{n}(z,z_{\set{n}\setminus\{i\}})\int\limits^z\!\!\tfrac{d}{d \tau}\big|_\zeta B(\cdot,z_i).
\end{multline}
\end{theorem}

\begin{remark}
In the cases we apply this formula in the present paper in~\eqref{eq:var-of-y} or~\eqref{eq:omega-B-variation}, the curve is fixed so that the derivative over~$\tau$ has an invariant meaning and coincides with the derivative at fixed~$\zeta$ for any meromorphic function $\zeta$ independent of~$\tau$. By that reason we drop a reference to a choice of~$\zeta$ in~\eqref{eq:var-of-y} and~\eqref{eq:omega-B-variation}. Besides, since the Bergman kernel is also fixed, the terms on the second line in~\eqref{eq:xydef} do not contribute in these cases.
\end{remark}

\begin{proof}[Proof of Theorem~\ref{thm:kaz}]

Though an independent proof of formula~\eqref{eq:xydef} is given in~\cite{K}, let us explain how it 
can be obtained from 
the deformation formulas of~\cite{Eynard-Intersection,DOSS}.


First, let us prove the following formula:
\begin{multline}\label{eq:xydef-fixedx}
\frac{d}{d \tau}\Big|_x\omega^{(g)}_{n}(z_{\set n})=\sum_{j=1}^N\res\limits_{z=q_j}
\omega^{(g)}_{n+1}(z,z_{\set n})\int\limits^z\!\!\tfrac{d}{d \tau}\big|_x\bigl(y\,dx\bigr)
 \\+\sum_{i,j=1}^N\res_{z=q_i}\res_{\tilde z=q_j}
 \omega^{(g),{\rm st}}_{n,2}(z,\tilde z,z_{\set n})\int\limits^z\!\!\!\int\limits^{\tilde z}\!\!\tfrac{d}{d \tau}\big|_x B
 +\sum_{i=1}^n\sum_{j=1}^N\res\limits_{z=q_j}\omega^{(g)}_{n}(z,z_{\set{n}\setminus\{i\}})\int\limits^z\!\!\tfrac{d}{d \tau}\big|_x B(\cdot,z_i), 
\end{multline}
i.e. the special case of formula~\eqref{eq:xydef} when $x$ is not being varied.

Let us start with proving the first line of ~\eqref{eq:xydef-fixedx}, i.e. the case when just $y$ depends on the parameter $\tau$. For completeness, we give a direct proof, but it also can be deduced from the statements of~\cite[Theorem~4.1]{Eynard-Intersection} and~\cite[Theorem~3.7]{DOSS}.

 Recall that $\omega^{(g)}_{n}$ can be expressed as a sum over graphs as in~\cite[Section~4.5]{EO-1st}. Note that this formula for $\omega^{(g)}_{n}$ in terms of sums over graphs uses only local information, cf.~\cite{Eynard-Intersection,DOSS}. For brevity, we do not repeat this construction (of $\omega^{(g)}_{n}$ as a sum over graphs) here and just refer to the original paper.  Note that in our case out of all of the ingredients only
 \begin{equation}
 	K(z,w)\coloneqq \dfrac{\int_{\bar{w}}^w B(\cdot, z)}{2(y(w) - y(\bar{w}))dx(w)} 
 \end{equation} 
 depends on $\tau$, where $\bar{w}$ corresponds to taking the local involution near one of the critical points of $x$ such that $x(\bar w) = x(w)$. Note that
\begin{align}
	\dtau K(z, w) &= - K(z, w) \frac{\dot{y}(w) - \dot{y}(\bar w)}{y(w) - y(\bar w)},
\end{align}
where $\dot{y}(w) = \dtau y(w)$. Consider
\begin{equation}
	\sum_{j=1}^N\res\limits_{z=q_j}
	\omega^{(g)}_{n+1}(z,z_{\set n})\int\limits^z\!\!\dot{y}\,dx.
\end{equation}
In the Eynard--Orantin sum over graphs for $\omega^{(g)}_{n+1}(z,z_{\set n})$ the variable $z$ appears on exactly one leaf of each graph, as part of 
\begin{equation}	
\res_{w_2=\vec q}K(w_1,w_2)B(w_2,z)f(\bar w_2)
\end{equation}
 (if the leaf is attached as the left child of the respective vertex, and with $w_2$ and $\bar w_2$ interchanged between the arguments of $B$ and $f$ otherwise), where $f(\bar w_2)$ corresponds to the contribution of the part of the subgraph growing out of the other child of the vertex labeled with $w_2$, and $\res_{w_2=\vec q}$ stands for $\sum_{j=1}^N\res_{w_2=q_j}$. We have:
\begin{align} \label{eq:resexch}
	\res\limits_{z=\vec q}\res\limits_{w_2=\vec q}K(w_1,w_2)B(w_2,z)f(\bar w_2)\int\limits^z\!\!\dot{y}\,dx & =	\res\limits_{w_2=\vec q}\res\limits_{z=\vec q}K(w_1,w_2)B(w_2,z)f(\bar w_2)\int\limits^z\!\!\dot{y}\,dx \\ \notag &\quad  + 	\res\limits_{w_2=\vec q}\res\limits_{z=w_2}K(w_1,w_2)B(w_2,z)f(\bar w_2)\int\limits^z\!\!\dot{y}\,dx,
\end{align}
cf.~\cite[Eq.~(3.43)]{DOSS}. Note that the first term in the right hand side of~\eqref{eq:resexch} vanishes as $B(w_2,z)\int^z\!\!\dot{y}\,dx$ is regular in $z$ at the critical points of $x$. Since $B$ has a double pole at the diagonal with a bi-residue $1$ and no other poles, the residue with respect to $z$ can be taken and the right hand side of~\eqref{eq:resexch} is thus equal to:
\begin{equation}\label{eq:Kfdydx}
		\res\limits_{w_2=\vec q}K(w_1,w_2)f(\bar w_2)\dot{y}(w_2) dx(w_2) .
\end{equation}
This can be nonzero only if the other child of the $w_2$ vertex is an arrowed edge, otherwise the dependence on $w_2$ is regular at the critical points of $x$ and the residue vanishes. Thus~\eqref{eq:resexch} is equivalent to
\begin{align}\label{eq:KdydxK}
	&\res\limits_{w_2=\vec q}\res\limits_{w_3=\vec q}K(w_1,w_2)\dot{y}(w_2) dx(w_2) K(\bar w_2,w_3) g(w_3) \\ \notag
	&=\res\limits_{w_2=\vec q}\res\limits_{w_3=\vec q} \dfrac{\int_{\bar{w_2}}^{w_2} B(\cdot, w_1)}{2(y(w_2) - y(\bar{w_2}))dx(w_2)}\, \dot{y}(w_2) dx(w_2)  \, \dfrac{\int_{\bar{w_3}}^{w_3} B(\cdot, \bar w_2)}{2(y(w_3) - y(\bar{w_3}))dx(w_3)}\,g(w_3),
\end{align}
where $f(\bar w_2) = \res_{w_3=\vec q}K(\bar w_2, w_3) g(w_3)$.
Again using the residue exchange formula (taking into account that $\int_{\bar{w_3}}^{w_3} B(\cdot, w_2)$ has poles at $w_2=w_3$ and $w_2=\bar w_3$, cf.~\cite[Eq.~(3.19)]{DOSS}), we get that~\eqref{eq:KdydxK} is equal to
\begin{align}\label{eq:w3w2threeterms}
	&\res\limits_{w_3=\vec q}\res\limits_{w_2=\vec q} \dfrac{\int_{\bar{w_2}}^{w_2} B(\cdot, w_1)}{2(y(w_2) - y(\bar{w_2}))}\, \dot{y}(w_2)  \, \dfrac{\int_{\bar{w_3}}^{w_3} B(\cdot, \bar w_2)}{2(y(w_3) - y(\bar{w_3}))dx(w_3)}\,g(w_3) \\ \notag
	&+\res\limits_{w_3=\vec q}\res\limits_{w_2=w_3} \dfrac{\int_{\bar{w_2}}^{w_2} B(\cdot, w_1)}{2(y(w_2) - y(\bar{w_2}))}\, \dot{y}(w_2)  \, \dfrac{\int_{\bar{w_3}}^{w_3} B(\cdot, \bar w_2)}{2(y(w_3) - y(\bar{w_3}))dx(w_3)}\,g(w_3) \\ \notag
	&+\res\limits_{w_3=\vec q}\res\limits_{w_2=\bar w_3} \dfrac{\int_{\bar{w_2}}^{w_2} B(\cdot, w_1)}{2(y(w_2) - y(\bar{w_2}))}\, \dot{y}(w_2)   \, \dfrac{\int_{\bar{w_3}}^{w_3} B(\cdot, \bar w_2)}{2(y(w_3) - y(\bar{w_3}))dx(w_3)}\,g(w_3)
\end{align}
The first term in~\eqref{eq:w3w2threeterms} vanishes since the expression under the residue is regular in $w_2$ at the critical points of $x$ (since the pole coming from the first denominator cancels with the zero coming from the first numerator). Again, since $B$ has a double pole at the diagonal with biresidue $1$, we can take the residue with respect to $w_2$ in the second and third terms of~\eqref{eq:w3w2threeterms}, getting
\begin{align} \label{eq:longdyK}
	&-\dfrac{1}{2}\res\limits_{w_3=\vec q} \dfrac{\dot{y}(w_3)}{y(w_3) - y(\bar{w_3})}\,   \dfrac{\int_{\bar{w_3}}^{w_3} B(\cdot,  w_1)}{2(y(w_3) - y(\bar{w_3}))dx(w_3)}\,g(w_3) \\ \notag
	&\qquad +\dfrac{1}{2}\res\limits_{w_3=\vec q} \dfrac{\dot{y}(\bar w_3)}{y(\bar w_3) - y(w_3)}  \, \dfrac{\int_{{w_3}}^{\bar w_3} B(\cdot, w_1)}{2(y(w_3) - y(\bar{w_3}))dx(w_3)}\,g(w_3) \\ \notag
	&=-\dfrac{1}{2}\res\limits_{w_3=\vec q} \dfrac{\dot{y}(w_3)-\dot{y}(\bar w_3)}{y(w_3) - y(\bar{w_3})}\,   \dfrac{\int_{\bar{w_3}}^{w_3} B(\cdot, w_1)}{2(y(w_3) - y(\bar{w_3}))dx(w_3)}\,g(w_3) \\ \notag
	&=-\dfrac{1}{2}\res\limits_{w_3=\vec q} \dfrac{\dot{y}(w_3)-\dot{y}(\bar w_3)}{y(w_3) - y(\bar{w_3})}\,  K(w_1,w_3) \,g(w_3) \\ \notag
	&=\dfrac{1}{2}\res\limits_{w_3=\vec q} \, g(w_3)  \dtau K(w_1,w_3)
\end{align}
(we used that $\bar w_2 = w_3$ if and only if $w_2=\bar w_3$). Analogously, we get exactly the same expression for the case when the $z$-leaf in a graph in the graph sum for $\omega^{(g)}_{n+1}$ was the right child. Taking all of this into account (as well as the remark after formula~\eqref{eq:Kfdydx}) and comparing the graph sums for $\dtau \omega^{(g)}_{n}(z_{\set n})$ (where we use the Leibniz rule to take the derivative of each summand with respect to $\tau$) and for $\res_{z=\vec q}\omega^{(g)}_{n+1}(z,z_{\set n})\int^z\!\!\dot{y}\,dx$, we conclude that they coincide. The coefficient $\frac{1}{2}$ which we got in the bottom line of~\eqref{eq:longdyK} cancels since for each $K$ in each graph in the graph sum for $\omega^{(g)}_{n}(z_{\set n})$ on which the derivative falls via the Leibniz rule we have exactly two corresponding graphs in the graph sum for $\res_{z=\vec q}\omega^{(g)}_{n+1}(z,z_{\set n})\int^z\!\!\dot{y}\,dx$, which are obtained by inserting a new vertex in the middle of this arrowed edge making two arrowed edges out of it, and attaching the $z$-leaf to that vertex, which can be done in two ways: as the left and the right child. This proves~\eqref{eq:xydef-fixedx} for the case when only $y$ depends on $\tau$.

The two sums in the second line of~\eqref{eq:xydef-fixedx}  (i.e. the case when $B$ depends on $\tau$ but $y$ does not) can be obtained from~\cite[Proof of Theorem 3.7]{DOSS}; specifically, from~\cite[Eqs.~(3.50)--(3.51)]{DOSS} (see also ~\cite{Eynard-Intersection}). Note that while~\cite[Eqs.~(3.50)--(3.51)]{DOSS} are written for local topological recursion, they are actually completely general and work for global spectral curves as well; one just needs to remove all the indices $i(v)$ and replace $\mathop{\res_{Z_h=0}}$ with $\sum_{i=1}^N\mathop{\res_{Z_h=q_i}}$. To get the formula that we need, one has to replace $B_{\mathrm{KdV}}$ with $B(\tau_0)$ and $B_{\mathrm{reg}}$ with $B(\tau)-B(\tau_0)$ and then consider the linear term in the Taylor expansion of~\cite[Eq.~(3.51)]{DOSS} in $\tau$ at $\tau=\tau_0$. In place of $\omega^{\mathrm{KdV}}$ we will get just $\omega(\tau_0)$. Only terms with a single $B^{\mathrm{reg}}$ survive in the sum over graphs in~\cite[Eq.~(3.50)]{DOSS}. The terms where $B^{\mathrm{reg}}$ comes from the second line of~\cite[Eq.~(3.51)]{DOSS} correspond precisely to the first sum in the second line of~\eqref{eq:xydef-fixedx}, precisely with $\omega^{(g),{\rm st}}_{n,2}$ given by~\eqref{eq:omegast}: the graphs with a single vertex and a single edge looping from it back to itself correspond to the first term in~\eqref{eq:omegast}, while the graphs with two vertices and a single edge connecting them correspond to the sum in~\eqref{eq:omegast}. Finally,  taking into account that $B^{\mathrm{reg}}$ enters $B$ in the third line of~\cite[Eq.~(3.51)]{DOSS}, we get the second sum in the second line of~\eqref{eq:xydef-fixedx} from there.  

Now let us obtain~\eqref{eq:xydef} from~\eqref{eq:xydef-fixedx}. Passing from the derivative at fixed~$x$ to the derivative at fixed~$\zeta$ is governed by the following relations: for a given function~$f$, a meromorphic $1$-form $\lambda$, and an $n$-differential $\omega_n$, respectively, we have
\begin{align}
\tfrac{d}{d u}\big|_\zeta f&=\tfrac{d}{d u}\big|_xf+\dot x\frac{df}{dx}, \label{eq:dzetaf}
\\\tfrac{d}{d u}\big|_\zeta\lambda &=\tfrac{d}{d u}\big|_x\lambda+d\left(\dot x\frac{\lambda}{dx}\right), 
\\\tfrac{d}{d u}\big|_\zeta\omega_n &=\tfrac{d}{d u}\big|_x\omega_n+\sum_{i=1}^n d_i\left(\dot x_i\frac{\omega_n}{dx_i}\right), \label{eq:dzetaomega}
\end{align}
where $\dot x=\tfrac{d}{d u}\big|_\zeta x$. The first equation is the usual chain rule for partial derivatives, and the other two follow from the first one. Substituting these expressions to~\eqref{eq:xydef-fixedx}, we arrive, after simplifications, exactly to~\eqref{eq:xydef} (do not forget that one needs not only to use~\eqref{eq:dzetaomega} for $\omega^{(g)}_n$ and then substitute $\tfrac{d}{d u}\big|_x\omega_n$ from~\eqref{eq:xydef-fixedx}, but also to use~\eqref{eq:dzetaf} and ~\eqref{eq:dzetaomega} backwards to express $\tfrac{d}{d \tau}\big|_xy$ and $\tfrac{d}{d \tau}\big|_xB$ which appear there in terms of $\tfrac{d}{d \tau}\big|_\zeta y$ and $\tfrac{d}{d \tau}\big|_\zeta B$ respectively).
\end{proof}

\section*{Statements and declarations}

\begin{enumerate}
	\item The authors have no financial or non-financial interests that are directly or indirectly related to the work submitted for publication.
	\item There is no data collected, produced or analyzed within the research related to the paper.
\end{enumerate}

\printbibliography

\end{document}